\documentclass[11pt]{article}
\usepackage{fullpage}
\usepackage{authblk}

\usepackage{amsmath,amssymb,mathrsfs,mathtools,amsfonts,amsthm}
\usepackage[square,sort,comma,numbers]{natbib}
\allowdisplaybreaks

\newtheorem{theorem}{Theorem}
\newtheorem{lemma}{Lemma}

\usepackage{comment,color}

\usepackage{algpseudocode}

\begin{document}


\title{Smoothed and Average-case Approximation Ratios of Mechanisms: Beyond the Worst-case  Analysis}


\author[1]{Xiaotie Deng}
\author[1]{Yansong Gao}
\author[2]{Jie Zhang\thanks{jie.zhang@soton.ac.uk}}
\affil[1]{Shanghai Jiao Tong University, China}
\affil[2]{University of Southampton, U.K.}
\date{}


\maketitle
\begin{abstract}
The approximation ratio has become one of the dominant measures in mechanism design problems. In light of analysis of algorithms, we define the \emph{smoothed approximation ratio} to compare the performance of the optimal mechanism and a truthful mechanism when the inputs are subject to random perturbations of the worst-case inputs, and define the \emph{average-case approximation ratio} to compare the performance of these two mechanisms when the inputs follow a distribution. 
For the one-sided matching problem, \citet{FFZ:14} show that, amongst all truthful mechanisms, \emph{random priority} achieves the tight approximation ratio bound of $\Theta(\sqrt{n})$.
We prove that, despite of this worst-case bound, random priority has a \emph{constant smoothed approximation ratio}.
This is, to our limited knowledge, the first work that asymptotically differentiates the smoothed approximation ratio from the worst-case approximation ratio for mechanism design problems.
For the average-case, we show that our approximation ratio can be improved to $1+e$.  
These results partially explain why random priority has been successfully  used in practice, although in the worst case the optimal social welfare is  $\Theta(\sqrt{n})$ times of what random priority achieves. 
These results also pave the way for further studies of smoothed and average-case analysis for approximate mechanism design problems, beyond the worst-case analysis.
 \end{abstract}

\section{Introduction}
Algorithmic mechanism design \citep{NR99, DBLP:journals/geb/NisanR01} deals with optimization problems where the input is provided by self-interested agents that participate in the mechanism by reporting their private information. If it best serves their purpose, they might have incentives to report false information. The goal of the designer is to motivate agents to always report truthfully. At the same time, the mechanism designer aims to optimize some objective function over the agents' reports, subject to a polynomial-time implementability constraint. Examples of applications include scheduling problems~\citep{DBLP:journals/mst/Koutsoupias14, DBLP:journals/algorithmica/ChristodoulouKV09, DBLP:conf/mfcs/KoutsoupiasV07}, facility location problems \citep{PT:09,DBLP:conf/aaai/Filos-RatsikasL15}, kidney exchange problems~\citep{ashlagi2010mix}, assignment problems~\citep{DG:10}, one-sided matching \citep{FFZ:14}, resource allocation \citep{guo2010strategyproof}, and auction design \citep{DBLP:conf/sigecom/HartlineR09,DBLP:conf/focs/AlaeiHNPY15,DBLP:journals/jacm/LehmannOS02,DBLP:journals/geb/MualemN08,DBLP:conf/soda/ArcherPTT03}. For a more detailed investigation, we refer the reader to \citep{NRTV07}. The canonical measure of evaluating how well a truthful mechanism approximately optimizes the objective is the \emph{approximation ratio} \citep{PT:09}. Given any instance, the approximation ratio compares the performance of the optimal mechanism (which always outputs an optimal solution but is not necessarily truthful) against the performance of a truthful mechanism. The worst-case ratio is the largest value of this ratio, amongst all possible inputs.


The difference of mechanism design to algorithm design is the additional constraint of motivating agents to act truthfully in the mechanism. In algorithm design problems, in contrast, the inputs are not controlled by rational agents, and the \emph{worst-case time complexity} is one of the dominant measures to evaluate the performance of algorithms. However, this is a very pessimistic measure. On one hand, if it is possible to  obtain a small worst-case complexity, then it is a very strong guarantee on the performance of the algorithm no matter what input is given. On the other hand, there are algorithms that perform well in practice but have a high worst-case complexity bound.
To address this disparity, the \emph{average-case} time complexity of an algorithm is an alternative measure to the worst-case complexity; it measures the time complexity of the algorithm, averaged over all possible inputs when they follow a certain distribution.  The main motivation of studying average-case complexity is that some algorithms may have to run for a high-order polynomial-time or even exponential time in the worst case, but the input for this to happen  may rarely or never occur in practice. So instead of only focusing on  unrealistic worst-case instances, researchers consider the performance of the algorithm on average. One criticism to average-case complexity, however, is that it requires assumptions about the distribution of inputs, and these are not guaranteed to hold in practice. Even for the same algorithm, when it is applied to different application areas, the real-world distribution may vary.  In light of this, \emph{smoothed complexity} is a hybrid of the worst-case and average-case analysis that inherits advantages of both. Specifically, it measures the expected performance of algorithms under slight random perturbations of the worst-case inputs. If the smoothed complexity of an algorithm is low, then it is unlikely that the algorithm will take a long time to solve practical instances whose data are subject to slight noise and imprecisions.  Although average-case and smoothed analysis are usually more complex than the worst-case analysis, they provide different measures from the worst-case complexity.


In this paper,  we extend the classical worst-case approximation ratio analysis of mechanisms to the smoothed approximation ratio and average-case approximation ratio  analysis. \emph{Average-case approximation ratio},  on average, measures the performance of a truthful mechanism in approximately maximizing social welfare (or minimizing social cost) against the performance of the optimal mechanism; the \emph{smoothed approximation ratio} compares the performance of these two mechanisms when the inputs are subject to random perturbations of worst-case inputs. Therefore, the central questions in smoothed analysis and average-case analysis in mechanism design framework are:

\vspace{3mm}
\emph{Given a mechanism design problem, in case the worst-case approximation ratio is asymptotically large, are there any mechanisms that achieve a constant smoothed approximation ratio and a constant average-case approximation ratio? Given a concrete mechanism for the problem, does it have a  constant smoothed approximation ratio and a constant average-case approximation ratio?}
\vspace{2mm}

As a first step of extending the worst-case analysis to the smoothed and average-case analysis in mechanism design, we study the problem of approximate social welfare maximization (without money) in the \emph{one-sided matching settings} (also referred to as the  \emph{house allocation problem}).  These settings consider the fundamental resource allocation problem of assigning items to agents, such that each agent receives exactly one item. It has numerous applications, such as assigning workers to shifts, students to courses, and patients to doctor appointments. In this problem, agents are asked to provide their preferences over items. In game-theoretic terms, these are the agents' von Neumann-Morgenstern utilities~\citep{vNM1953,von2007theory}. Social welfare is the sum of all agents' utilities. 
It is easy to see that agents, as self-interested identities, have an incentive to misreport their preferences if they can benefit from this behavior. The \emph{random priority} mechanism, apart from being truthful, also satisfies the desirable properties of anonymity and ex-post Pareto efficiency.  In term of social welfare maximization, \citet{FFZ:14} show that amongst all truthful mechanisms, random priority achieves the worst-case approximation ratio  tight bound of $\Theta(\sqrt{n})$. That is, random priority can guarantee an upper bound of $O(\sqrt{n})$ and there is a worst-case instance on which no mechanism can do better than $\Omega(\sqrt{n})$. Nevertheless, the tight bound instance has a very unique structure such that it is very unlikely to happen in practice. Therefore, we are in great request to  understand  how well  random priority  performs on average, when the instances are sampled from a certain distribution. How well it performs under some random perturbations of the worst-case inputs? We address these questions in this paper.



\subsection{Our contribution}

To the best of our knowledge, this is the first work that asymptotically differentiates the smoothed approximation ratio from the worst-case approximation ratio for  mechanism design problems. In particular, we show the following results:
\begin{itemize}
\item The random priority mechanism has a constant smoothed approximation ratio.
\item The average-case approximation ratio of random priority is upper bounded by a constant $1+e$, when agents' valuations are drawn from the uniform distribution $\mathrm{U}(0,1)$.
\end{itemize}
\noindent
Our results imply that, although in the worst-case the optimal social welfare is $\Theta(\sqrt{n})$ times of the social welfare attainable by random priority, under polynomial small perturbation around the worst-case inputs and on average, random priority achieves a constant factor of the optimal social welfare.

In \cite{FFZ:14}, the tight bound examples for the worst-case approximation ratio have a unique structure where the preferences of all agents over the items have the same ordering, and the values are all close to either 1 or 0. From the average-case perspective, these examples rarely happen if the valuations are independently and identically drawn from a uniform distribution. From the smoothed analysis perspective, this unanimous ordering would break up after any random perturbation. This is the high-level intuition behind why the smoothed and average-case approximation ratios could be asymptotically different from the worst-case approximation ratio.


\subsection{Related work}

\citet{DBLP:conf/stoc/SpielmanT01,DBLP:journals/mp/SpielmanT03} first propose the methodology of smoothed analysis of algorithms with the attempt to explain why the simplex algorithms usually run in polynomial time in practice. They start with the shadow-vertex pivot rule and show its polynomial smoothed complexity. Since then, smoothed analysis has been studied on a variety of different problems and algorithms, including linear programming \citep{DBLP:conf/soda/BlumD02,DBLP:journals/siammax/SankarST06,DBLP:journals/mp/SpielmanT03}, online and approximation algorithms \citep{DBLP:conf/dagstuhl/BecchettiLMSV05,DBLP:journals/algorithmica/BlaserMR13,DBLP:journals/tcs/SchaferS05}, searching and sorting \citep{DBLP:conf/mfcs/BanderierBM03,DBLP:journals/algorithmica/FouzKMJ12,DBLP:journals/tcs/MantheyR07}, game theory \citep{DBLP:conf/icalp/BorosEFGMM11,DBLP:journals/jacm/ChenDT09},  local search \citep{DBLP:journals/siamcomp/ArthurV09,DBLP:journals/algorithmica/EnglertRV14}, clustering and knapsack problem \citep{DBLP:journals/it/MantheyR11}. A comprehensive survey can be found in \citep{DBLP:journals/cacm/SpielmanT09}.

For average-case analysis of algorithms, different results have been obtained in, for example,  quicksort for sorting problem \citep{DBLP:books/daglib/0023376} and the simplex algorithm for solving linear programming \citep{DBLP:journals/mmor/Borgwardt82}. We refer the reader to \citep{Szpankowsk} for a comprehensive survey.

In the presence of incentives, the one-sided matching problem  was originally defined by  \citet{HZ:79}, and has been studied extensively ever since \cite{ZHOU:90, BM:01, SVE:99, DG:10,mennle2014axiomatic}.  We refer the interested reader to  surveys on this problem \citep{AS:13,SU:11}. The \emph{random priority} mechanism, also known as \emph{random serial dictatorship}, has been extensively studied  \cite{BM:01, abdulkadirouglu1998random, Sv99}. It has also been widely used in practice, for example, the supplementary round of school student assignment mechanism in New York City, is shown to be equivalent to a random priority mechanism \cite{PathakSethuraman}.

In the Bayesian auction design literature \cite{DBLP:journals/sigecom/ChawlaS14, DBLP:conf/stoc/HartlineL10}, the focus is on how well a truthful mechanism can approximately maximize the expected revenue, when instances are taken from the entire input space. More specifically, the dominant approach in the study of Bayesian auction design is the \emph{ratio of expectations}. One disadvantage of this approach is that it does not directly compare the performance of the two mechanisms on specific inputs. To address this problem, in what follows, we present a different approach, the \emph{expectation of the ratio}, and  compare them in more detail.


\section{Preliminaries}

We study the one-sided matching problem that consists of $n$ agents and $n$ indivisible  items. 
All the agents are endowed with von Neumann - Morgenstern utilities over the items. The VNM utility theorem states that any rational agent whose preference  satisfies four axioms, namely \emph{completeness}, \emph{transitivity}, \emph{continuity}, and \emph{independence}, is endowed with a utility function to represent its preference. These agents  report their preferences to a mediator; based on their report, the problem is to allocate items to agents, according to a random permutation such that every agent receives exactly one item.

In this paper, we adopt the \emph{unit-range}  canonical representation of agents' valuation \cite{Barbera10, ZHOU:90}. 
That is, for  the utility $a_{ij}$ of agent $i$ receiving item $j$, we have $\max_j \{a_{ij}\}=1, \min_j \{a_{ij}\}=0$, $\forall i\in [n]$ \footnote{We note here that our model would be more general and some calculations would be simpler if we drop the constraint $\max_j \{a_{ij}\}=1$ and $\min_j \{a_{ij}\}=0$, but only require that $0\le a_{ij} \le 1$, $\forall i, j \in [n]$. The only reason to have such a constraint is to follow the \emph{unit-range}  canonical representation of agents' valuation studied in literature.}.  Following this, a valuation profile (or instance) of agents' preferences can be represented by a matrix $\mathrm{A}=[a_{ij}]_{n \times n}$, where row vector $\mathbf{a_i}=(a_{i1},\ldots,a_{in})$ indicates the valuation of agent $i$'s preference.  An \emph{allocation} is  an assignment of items to agents. We denote an allocation by a matrix $\mathrm{X}=[x_{ij}]_{n \times n}$, where $x_{ij}$ indicates the probability of  agent $i$ receiving item $j$. Given any preferences of the agents as input, a \emph{mechanism} is a mapping from which to an allocation $\mathrm{X}$ as output.


We denote  the set of all possible instances by $\mathcal{A}$ and denote the set of all possible allocation by $\mathcal{X}$. Given a mechanism $\mathrm{M}$ and a valuation profile $\mathrm{A}\in \mathcal{A}$, as well as its allocation $\mathrm{X}(\mathrm{A})\in \mathcal{X}$, we denote the utility of agent $i$ by  $u_i (\mathrm{X}(\mathrm{A})) = \sum_{j} a_{ij}x_{ij}$ and denote the social welfare by $SW_{\mathrm{M}}(\mathrm{X}(\mathrm{A})) = \sum_{i} u_i(\mathrm{X}(\mathrm{A}))$. When the context is clear, we drop the allocation notation and simplely refer them by $u_i(\mathrm{A})$ and $SW_{\mathrm{M}}(\mathrm{A})$.  We note that there is another interpretation of our one-sided matching problem: items are divisible and $x_{ij}$ is the fraction of agent $i$ receiving item $j$. Since the number of agents is equal to the number of items, and every agent receives exactly one item, the allocation matrix $\mathrm{X}$ is a doubly stochastic matrix, i.e., $\sum_{j} x_{ij}=1, \forall i$ and $\sum_{i} x_{ij}=1, \forall j$. According to the Birkhoff - von Neumann theorem, every doubly stochastic matrix can be decomposed into a convex combination of some permutation matrices. Therefore, $u_i(\mathrm{A})$ and $SW_{\mathrm{M}}(\mathrm{A})$ can be interpreted as expected utilities and expected social welfare in the indivisible items setting, and can be interpreted as exact utilities and exact social welfare in the divisible items setting.

Agents are self-interested and look to maximize their utilities by giving a mendacious preference to the mechanism as part of the input. In approximate mechanism design, we restrict our interest to the class of \emph{truthful mechanisms}, i.e., the mechanisms in which agents cannot improve their utilities by misreporting. The canonical measure of efficiency of a truthful mechanism $\mathrm{M}$ is the \emph{worst-case approximation ratio},
\begin{equation*}
r_{\text{worst}}(\mathrm{M}) = \sup_{\mathrm{A} \in \mathcal{A}} \frac{SW_{\mathrm{OPT}}(\mathrm{A})}{SW_{\mathrm{M}}(\mathrm{A})},
\end{equation*}
where $SW_{\mathrm{OPT}}(\mathrm{A})= \max_{\mathrm{X} \in \mathcal{X}}\sum_{i=1}^{n}u_i(\mathrm{X})$ is the optimal social welfare which is equal to the value of the maximum weight matching between agents and items. This ratio compares social welfare of the optimal allocation against social welfare of a truthful mechanism $\mathrm{M}$. Note that  the ratio is no less than 1. 

Random priority mechanism fixes an ordering of the agents uniformly at random and then lets them pick their most preferred item from the set of available items based on this ordering. It is shown in \citep{FFZ:14} that random priority achieves the matching approximation ratio bound of $\Theta(\sqrt{n})$. This result implies that random priority is asymptotically the best truthful mechanism. 

\subsection{Smoothed and average-case approximation ratios}

Analogously to the definition of smoothed complexity of algorithms, we define the  \emph{smoothed approximation ratio}  as follows:
\begin{equation*}
r_{\text{smoothed}}(\mathrm{M}) = \max_{\mathrm{A}} \mathop{\mathbb{E}}_{g_{ij} \sim \mathrm{N}(0,1)} \left[ \frac{SW_{\mathrm{OPT}}(\mathrm{A}+\sigma||\mathrm{A}||\mathrm{G})}{SW_{\mathrm{M}}(\mathrm{A}+\sigma||\mathrm{A}||\mathrm{G})} \right] ,
\end{equation*}
where the parameter $\sigma$ is the size of the perturbation and $\mathrm{G}=[g_{ij}]_{n \times n}$. That is to say, $\sigma||\mathrm{A}||\mathrm{G}$ is a matrix of independent Gaussian variables of mean 0 and standard deviation $\sigma||\mathrm{A}||$. We multiply by $||\mathrm{A}||$ to relate the magnitude of the perturbation to the magnitude of the input $\mathrm{A}$. 
We say that a mechanism has a \emph{constant smoothed approximation ratio} if its smoothed approximation ratio is polynomial only in $1/ \sigma$, when the size of the input $n$ approaches infinity \footnote{We will show that the smoothed approximation ratio of random priority is polynomial only in $1/\sigma$ but not in $n$. This is in contrast to its worst-case approximation ratio $\Theta({\sqrt{n}})$. For this reason our result asymptotically differentiates the smoothed approximation ratio from the worst-case ratio.  It is a scenery analogously to the analysis of algorithms, where in some problems the smoothed complexity is polynomial in $n$ and $1/\sigma$ but the worst-case complexity is exponential in $n$. In both studies, the point of smoothed analysis is to show that although the mechanism (algorithm) may perform poorly in the worst case, it performs well under slight random perturbations of worst-case inputs.}.
Following the natural of  the unit-range representation, we consider the \emph{property-preserving perturbation} by restricting a natural perturbation model to preserve certain properties of the input. Specifically, the perturbation preserves agents' 0 and 1 valuation and all valuations stay in the interval $[0,1]$ \footnote{We note that the choice of Gaussian perturbation is standard in classical smoothed analysis of algorithms, and restriction of property-preserving perturbations is necessary and meaningful even in average-case analysis of algorithms. }.



Similarly, we define the \emph{average-case approximation ratio} of mechanisms as follows:
\begin{equation*}
r_{\text{average}}(\mathrm{M}) = \mathop{\mathop{\mathbb{E}}}_{a_{ij}\sim \mathrm{U}}\left[ \frac{SW_{\mathrm{OPT}}(\mathrm{A})}{SW_{\mathrm{M}}(\mathrm{A})} \right] ,
\end{equation*}
where the elements $a_{ij}$ of input $\mathrm{A}$ is chosen from a distribution $\mathrm{U}$. In this paper, we consider the case that agent's values $a_{ij}$ are independent variables and follow the uniform distribution $\mathrm{U}(0,1)$. For any agent $i$, since $a_{ij}, j=1,\ldots,n$, are independent and identically distributed $\mathrm{U}(0, 1)$ random variables,
the sum $\sum_{j=1}^{n}a_{ij}$ follows the {\em Irwin-Hall distribution}. So,
\begin{equation*}
\Pr\left[\sum\limits_{j=1}^{n}a_{ij}\le x\right] = \frac{1}{n!}\sum_{k=0}^{\lfloor x-1\rfloor}(-1)^{k}\binom{n}{k}(x-1-k)^{n} ,
\end{equation*}
where $\lfloor \cdot \rfloor$ is the floor function. 

In the following, we contrast this approach to the established literature on Bayesian mechanism design.

\subsection{Bayesian mechanism design approach}

In Bayesian mechanism design, there is also a prior distribution from which the agent types come from, but the objective is to characterize the maximum ratio (for some given distribution of the agent types) of the expected social welfare of a truthful mechanism over the expected social welfare of the optimal mechanism, i.e., the ratio of expectations. That is, the objective is to characterize the ratio $r$ in the following formula,
\begin{equation*}
\mathop{\mathop{\mathbb{E}}} \left[ SW_{\mathrm{OPT}}(\mathrm{A}) \right] \le r \cdot \mathop{\mathop{\mathbb{E}}} \left[ {SW_{\mathrm{M}}(\mathrm{A})} \right] .
\end{equation*}

Thus, the fundamental difference to our approach is that this measurement does not directly compare the performance of the two mechanisms on specific inputs. We discuss this in more detail in the following.

\subsection{Comparison of the two approaches}

We note that the worst-case approximation ratio compares the performance of the two mechanisms on a \emph{case-by-case basis}. In addition, the smoothed analysis of algorithms is defined as the performance of the algorithm on the worst-case input when it is subject to a slight random perturbation. Therefore, it measures the averaged performance of the algorithm in a small neighbourhood area of \emph{an individual input}. In essence, both metrics consider a specific input (or some small noise around a specific input). Indeed, our approach is informed by these two metrics.

Average-case approximation ratio and the Bayesian approach each has their strength in measuring how good a truthful mechanism approximates the optimal mechanism. The random priority mechanism for the one-sided matching problem is used for example, in the supplementary round of school student assignment in the New York City once in a while. In this case, we are interested in how likely a truthful mechanism performs well on a particular instance. Therefore, the average-case approximation ratio, which is the expectation of ratio, fits in this need.  On the other hand, online auction mechanisms, as the revenue source for Internet businesses, are used one a daily base. In that case, we are interested in how well a truthful mechanism performs comparing to the optimal mechanism when instances are sampled throughout the entire input space. Therefore, the Bayesian approach, which is the ratio of expectations, is more suitable.

We note that our definitions, the expectation of ratio, introduce more technical challenges. As we are interested in characterizing the expectation of a non-linear function (ratio of two variables), we can no longer handle expectations of two different variables separately. 
Let us take for example when agents' valuations are drawn independently and identically from the uniform distribution $\mathrm{U}(0,1)$. In the Bayesian approach, it is trivial to see that the expected social welfare of the completely random mechanism (by ignoring agents' valuations and allocate items totally random) is $n/2$, and the expected social welfare of the optimal mechanism is less than $n$. So the ratio of two expectations is less than 2. However, for the case-by-case comparison in the average-case approximation ratio, it is not directly clear what the ratio would be. In this paper, as a first step, we study the uniform distribution for the average-case ratio and the smoothed ratio. We suspect that different distributions would result in different ratios. In contrast, by the linearity of the ratio in two expectations, in the well-studied Bayesian auction design literature, it is common that positive results hold for a class of distributions, such as monotone hazard distributions \cite{DBLP:conf/stoc/HartlineL10}. 



\section{Smoothed Analysis}\label{sec:smoothed}

In this section we show that random priority has a constant smoothed approximation ratio. 

Firstly, let us understand the structure of the profile space $\mathcal{A}$ in the unit-range setting. For each agent $i$ with valuations $\mathbf a_i=(a_{i1},\ldots,a_{in})$, it has utility 0 on its least preferred item and has utility 1 on its most preferred item.  So there remains $n-2$ elements which are random variables following the uniform distribution $\mathrm{U}(0,1)$. 
Note that there are $n(n-1)$ possible choices to select $n-2$ out of $n$ elements to be random variables, while the other two elements being either 0 or 1. So there are $(n(n-1))^n$ possible configuration of such random variables. Therefore, we can decompose the sample space $\mathcal{A}$ into  sets $\mathrm{S_k}, k=1,\ldots, (n(n-1))^n$, where each set $\mathrm{S_k}$ contains the instances $\mathrm{A}$ that the specific two out of $n$ values of $a_{ij}, j\in [n]$ are fixed for each agent $i$. 

Given any valuation profile $\mathrm{A}$, its 0 and 1 entries are fixed so it must belong to exactly one of the set $\mathrm{S}_k$. Let $\mathrm{S}(\mathrm{A})\in\{\mathrm{S}_{k}\}_{k\in [(n(n-1))^{n}]}$ denote the set where $\mathrm{A}$ belongs to. Let $\mathrm{H}  (\mathrm{A}) = \{ (i,j) : a_{ij}=0 \ \text{or} \ a_{ij}=1 \}$,  clearly $\mathrm{H}  (\mathrm{A})$ has $2n$ elements. Given any set $\mathrm{S_k}$ and any two instances $\mathrm{A}_1$ and $\mathrm{A}_2$ chosen from  $\mathrm{S_k}$, it is easy to see that $\mathrm{H}  (\mathrm{A}_1) = \mathrm{H}  (\mathrm{A}_2)$. 

For convenience, denote $\Gamma=(\gamma_{ij})_{n\times n}=\sigma||\mathrm{A}||\mathrm{G}$. The probability density function of $\Gamma$ is
\begin{equation*}
f(\Gamma) =  K \cdot e^{-  \frac{|\Gamma|^2}{2{\sigma}^2 ||\mathrm{A}||^2}}
= K \cdot e^{- \sum\limits_{i=1}^{n}\sum\limits_{j=1}^n \frac{\gamma_{ij}^2}{2{\sigma}^2 ||\mathrm{A}||^2}} ,
\end{equation*}
where $K$ is a constant, $|\Gamma|^2 = \sum\limits_{i=1}^{n}\sum\limits_{j=1}^n \gamma_{ij}^2$. 
Therefore, our problem of characterizing the smoothed approximation ratio of random priority is reduced to computing the upper bound of
\begin{equation*}
\max_{\mathrm{A}\in \mathcal{A}}  \int_{\mathrm{A}+\Gamma \in \mathrm{S}(\mathrm{A})} \frac{1}{K'} \cdot e^{- \frac{|\Gamma|^2}{2{\sigma}^2 ||\mathrm{A}||^2}} \cdot \frac{SW_{OPT}(\mathrm{A}+\Gamma)}{SW_{RP}(\mathrm{A}+\Gamma)} d\Gamma ,
\end{equation*}
where $K' = \int_{\mathrm{A}+\Gamma \in \mathrm{S}(\mathrm{A})} e^{- \frac{|\Gamma|^2}{2{\sigma}^2 ||\mathrm{A}||^2}} d\Gamma$. As we clarified in the Preliminaries, we consider the property-preserving perturbation that preserves agents' 0 and 1 valuation and all valuations stay in the interval $[0,1]$. Therefore, our integral is taken over the space $\mathrm{A}+\Gamma \in \mathrm{S}(\mathrm{A})$. We would need the following auxiliary lemma. 

\noindent
\begin{lemma}\label{lem:matrixnorm}
For any valuation profile $\mathrm{A} \in \mathcal{A}$ and standard vector norm such as Euclidean norm, $p$-norm and maximum norm, we have $||\mathrm{A}|| \ge 1$.
\end{lemma}

\begin{proof}
In the unit-range setting, we have $a_{ij}\ge 0$ and $\max_j\{a_{ij}\}=1$. Without loss of generality, suppose $a_{st}=1$ for some $s\in [n]$ and $t\in [n]$, then $||\mathbf{A}||_{2}=\sqrt{\sum_{i=1,j=1}^{ n}|a_{ij}|^{2}}\ge\sqrt{a_{st}^{2}}=1$, $||\mathbf{A}||_{p}=\left( \sum_{i=1,j=1}^{n}|a_{ij}|^{p} \right)^{\frac{1}{p}}\ge ( |a_{st}|^{p})^{\frac{1}{p}}=1$, $||\mathbf{A}||_{\infty}=\max_{i\in [n],j\in [n]}|a_{ij}|\ge |a_{st}|=1.$
\end{proof}

Our main result of this section is the following.

\begin{theorem}\label{thm:smoothed}
Random priority has a constant smoothed approximation ratio. That is to say, its smoothed approximation ratio is polynomial in $1/\sigma$, when the input size $n$ approaches infinity.
\end{theorem}

Obviously, when $\frac{1}{\sigma} \ge \sqrt{n}$, since the smoothed approximation ratio is upper bounded by worst-case approximation ratio $O({\sqrt{n}})$, there $\exists M>0$, such that
\begin{equation*}
r_{\text{smoothed}} \le r_{\text{worst}} \le O({\sqrt{n}}) \le M \cdot \frac{1}{\sigma}
\end{equation*}

In the following we focus on the case when $\frac{1}{\sigma} < \sqrt{n}$. We further divide our analysis into two cases, depending on the size of the perturbation $\sigma$ and the magnitude of $\mathrm{A}$.


\indent

{\bf Case 1:} $\sigma ||\mathrm{A}|| \le 1$.

Firstly, the following lemma lower bounds the constant $K'$.
\begin{lemma}\label{lem:smoothed1}
When $\sigma ||\mathrm{A}|| \le 1$, for any instance $\mathrm{A}$ and Gaussian perturbation $\Gamma$, we have
\begin{equation*}
K' = \int_{\mathrm{A}+\Gamma \in \mathrm{S}(\mathrm{A})} e^{- \frac{|\Gamma|^2}{2{\sigma}^2 ||\mathrm{A}||^2}} d\Gamma
\ge  \left( e^{-\frac{1}{2}} \cdot \sigma ||\mathrm{A}|| \right)^{n(n-2)}
\end{equation*}
\end{lemma}
\begin{proof}
\begin{align*}
K' &= \int_{\mathrm{A}+\Gamma \in \mathrm{S}(\mathrm{A})} e^{- \frac{|\Gamma|^2}{2{\sigma}^2 ||\mathrm{A}||^2}} d\Gamma
= \prod_{\substack{i=1 \\ j=1}}^n \int_{\mathrm{A}+\Gamma \in \mathrm{S}(\mathrm{A})} e^{- \frac{|\gamma_{ij}|^2}{2{\sigma}^2 ||\mathrm{A}||^2}} d\gamma_{ij}  \\
&= \prod_{(i,j)\in \mathrm{H}(\mathrm{A}) }  e^{- \frac{0}{2{\sigma}^2 ||\mathrm{A}||^2}} \prod_{(i,j)\not\in \mathrm{H}(\mathrm{A}) }\int_{a_{ij}+\gamma_{ij} \in [0,1]} e^{- \frac{|\gamma_{ij}|^2}{2{\sigma}^2 ||\mathrm{A}||^2}} d\gamma_{ij}  \\
&= 1 \cdot  \prod_{(i,j)\not\in \mathrm{H}(\mathrm{A}) }\int_{-a_{ij}}^{1-a_{ij}} e^{- \frac{|\gamma_{ij}|^2}{2{\sigma}^2 ||\mathrm{A}||^2}} d\gamma_{ij}    
= \prod_{(i,j)\not\in \mathrm{H}(\mathrm{A}) }\left( \int_{0}^{1-a_{ij}} e^{- \frac{|\gamma_{ij}|^2}{2{\sigma}^2 ||\mathrm{A}||^2}} d\gamma_{ij} + \int_{0}^{a_{ij}} e^{- \frac{|\gamma_{ij}|^2}{2{\sigma}^2 ||\mathrm{A}||^2}} d\gamma_{ij} \right)   \\
&\ge  \prod_{(i,j)\not\in \mathrm{H}(\mathrm{A}) }\left( \int_{0}^{1-a_{ij}} e^{- \frac{|\gamma_{ij}|^2}{2{\sigma}^2 ||\mathrm{A}||^2}} d\gamma_{ij} + \int_{1-a_{ij}}^{1} e^{- \frac{|\gamma_{ij}|^2}{2{\sigma}^2 ||\mathrm{A}||^2}} d\gamma_{ij} \right)  
= \prod_{(i,j)\not\in \mathrm{H}(\mathrm{A}) }\int_{0}^{1} e^{- \frac{|\gamma_{ij}|^2}{2{\sigma}^2 ||\mathrm{A}||^2}} d\gamma_{ij}  \\
&= \left( \int_{0}^{1} e^{- \frac{|\gamma_{ij}|^2}{2{\sigma}^2 ||\mathrm{A}||^2}} d\gamma_{ij} \right)^{n(n-2)} 
\ge \left( \int_{0}^{\sigma ||\mathrm{A}||} e^{- \frac{|\gamma_{ij}|^2}{2{\sigma}^2 ||\mathrm{A}||^2}} d\gamma_{ij} \right)^{n(n-2)}  \\
&\ge \left( \int_{0}^{\sigma ||\mathrm{A}||} e^{- \frac{1}{2}} d\gamma_{ij} \right)^{n(n-2)}  
=  \left( e^{-\frac{1}{2}} \cdot \sigma ||\mathrm{A}|| \right)^{n(n-2)} 
\end{align*}
\end{proof}
\noindent
Secondly, given any instance $\mathrm{A}$ and its associated set $\mathrm{S}(\mathrm{A})$, we further partition the set $\mathrm{S}(\mathrm{A})$ into two subsets, according to the value of social welfare of random priority on its elements, namely $\mathrm{S}_{c}(\mathrm{A}) = \{ \mathrm{B} \in \mathrm{S}(\mathrm{A}): SW_{RP(\mathrm{B})} \le n^c \}$ and the residual $\mathrm{S}(\mathrm{A}) / \mathrm{S}_{c}(\mathrm{A}) = \{ \mathrm{B} \in \mathrm{S}(\mathrm{A}): SW_{RP(\mathrm{B})} > n^c \}$, where $0<c<1$ is a parameter that will facilitate us to prove our main theorem. Let $\mathrm{V(\mathrm{S}_{c}(\mathrm{A}))}$ and $\mathrm{V(\mathrm{S}(\mathrm{A}))}$ be the volume of $\mathrm{S}_{c}(\mathrm{A})$ and $\mathrm{S}(\mathrm{A})$, respectively. 
We upper bound the fraction of the elements of the set $\mathrm{S}(\mathrm{A})$ for which the social welfare of random priority on its elements is no more than $n^c$. A useful observation here is that, this objective is  equivalent to  computing the probability that the social welfare of random priority on any instance $\mathrm{A}$ is no more than $n^c$, when the agents' values $a_{ij}$ are independently and identically drawn from the uniform distribution $\mathrm{U}(0,1)$.

\begin{lemma}\label{lem:smoothed2}
Given any $n$, for every $0<c<1$,  we have
\begin{equation*}
\frac{\mathrm{V(\mathrm{S}_{c}(\mathrm{A}))}}{\mathrm{V(\mathrm{S}(\mathrm{A}))}} \le \frac{\mathrm{e}^{2n}}{\sqrt{2\pi}n} \cdot \left(\frac{2\mathrm{e}}{n^{1-c}}\right)^{n(n-2)}.
\end{equation*}
\end{lemma}

\begin{proof}
We note that random priority is a truthful mechanism, and it fixes an ordering of agents uniformly at random. Every agent $i$ has a probability of $1/n$ to be selected first to choose an item, a probability of $2/n$ to be selected first or second to choose an item, and so on. That is to say, for each agent $i$, if we sort the items in decreasing order according to agent $i$'s preference, then the allocation vector $\mathbf x_i=(x_{ij})$ obtained by random priority would stochastically dominates the vector $(\frac{1}{n},\ldots,\frac{1}{n})$. I.e., for every $k=1,\ldots,n$, it holds that $\sum_{j=1}^k x_{ij} \ge k/n$.  

Therefore, the utility of agent $i$ in any instance $\mathrm{A}$ is $u_{i}(\mathrm{A}) = \sum\nolimits_{j=1}^n a_{ij}x_{ij} \ge \frac{1}{n} \cdot \sum_{j=1}^{n}a_{ij}$. Hence, we have social welfare $\mathrm{SW}_{\mathrm{RP}}(\mathrm{A})=\sum_i u_i(\mathrm{A}) \ge \frac{1}{n} \sum_{i,j} a_{ij}$. So, $\mathrm{SW}_{\mathrm{RP}}(\mathrm{A}) \le n^{c}$ implies $\sum\nolimits_{i\in[n],j\in[n]} a_{ij}\le n^{1+c}$. Together with the observation noted above, we utilize the Irwin-Hall distribution to prove the lemma.
\begin{align*}
&\frac{\mathrm{V(\mathrm{S}_{c}(\mathrm{A}))}}{\mathrm{V(\mathrm{S}(\mathrm{A}))}} 
= \Pr\left[\mathrm{SW}_{\mathrm{RP}}(\mathrm{A}) \le n^{c}\right]
\le \Pr\left[ \sum\nolimits_{i\in[n],j\in[n]} a_{ij}\le n^{1+c}\right] 
= \Pr\left[ \sum\nolimits_{(i,j)\not\in \mathrm{H}  (\mathrm{A}) } a_{ij}\le n^{1+c} - n\right]  \\
&= \frac{1}{(n(n-2))!}\sum_{k=0}^{\lfloor n^{1+c}-n\rfloor}(-1)^{k}\binom{n(n-2)}{k}(n^{1+c}-n-k)^{n(n-2)} \\
&\le \frac{1}{(n(n-2))!}\sum_{k=0}^{n(n-2)}\binom{n(n-2)}{k}n^{(1+c)(n(n-2))}
 = \frac{1}{(n(n-2))!}\cdot (2n^{1+c})^{n(n-2)}  \\
&\le  \frac{(2n^{1+c})^{n(n-2)}}{\sqrt{2\pi(n(n-2))}\cdot \left(\frac{n(n-2)}{\mathrm{e}}\right)^{n(n-2)}}  
 =\frac{(2n^{1+c})^{n(n-2)}}{\sqrt{2\pi}n\cdot \left(\frac{n^2}{\mathrm{e}}\right)^{n(n-2)}}\cdot \frac{\sqrt{2\pi}n\cdot \left(\frac{n^2}{\mathrm{e}}\right)^{n(n-2)}}{\sqrt{2\pi(n(n-2))}\cdot \left(\frac{n(n-2)}{\mathrm{e}}\right)^{n(n-2)}}  \\
&= \frac{1}{\sqrt{2\pi}n}\cdot \left(\frac{2\mathrm{e}}{n^{1-c}}\right)^{n(n-2)}\cdot \left(1+\frac{2}{n-2}\right)^{n(n-2)+\frac{1}{2}}  \\
& \le \frac{1}{\sqrt{2\pi}n}\cdot \left(\frac{2\mathrm{e}}{n^{1-c}}\right)^{n(n-2)}\cdot \mathrm{e}^{2n}
 = \frac{\mathrm{e}^{2n}}{\sqrt{2\pi}n}\cdot \left(\frac{2\mathrm{e}}{n^{1-c}}\right)^{n(n-2)}
\end{align*}
\end{proof} 
\vspace{-3mm}
Combing these two lemmas we can prove our main theorem under this case. 

{\bf Proof of Case 1 of Theorem \ref{thm:smoothed}}:\\
Note that when $\mathrm{A}+\Gamma \not\in \mathrm{S}_{c}(\mathrm{A})$, we can upper bound $\frac{SW_{OPT(\mathrm{A}+\Gamma)}}{SW_{RP(\mathrm{A}+\Gamma)}} \le n^{1-c}$, and in all cases   $\frac{SW_{OPT(\mathrm{A}+\Gamma)}}{SW_{RP(\mathrm{A}+\Gamma)}} \le O{(\sqrt{n})}$ according to the worst-case approximation ratio result \cite{FFZ:14}.  In addition, by Lemma \ref{lem:smoothed2} we have,
\begin{equation*}
\mathrm{V(\mathrm{S}_{c}(\mathrm{A}))} \le \frac{\mathrm{e}^{2n}}{\sqrt{2\pi}n} \cdot \left(\frac{2\mathrm{e}}{n^{1-c}}\right)^{n(n-2)} \cdot \mathrm{V(\mathrm{S}(\mathrm{A}))} =  \frac{\mathrm{e}^{2n}}{\sqrt{2\pi}n} \cdot \left(\frac{2\mathrm{e}}{n^{1-c}}\right)^{n(n-2)} \cdot 1 
\end{equation*}
Therefore,
\begin{align*}
& \mathop{\mathbb{E}}_{g_{ij} \sim \mathrm{N}(0,1)} \left[ \frac{SW_{\mathrm{OPT}}(\mathrm{A}+\sigma||\mathrm{A}||\mathrm{G})}{SW_{\mathrm{RP}}(\mathrm{A}+\sigma||\mathrm{A}||\mathrm{G})} \right]  
=  \int_{\mathrm{A}+\Gamma \in \mathrm{S}(\mathrm{A})} \frac{e^{- \frac{|\Gamma|^2}{2{\sigma}^2 ||\mathrm{A}||^2}}}{\int_{\mathrm{A}+\Gamma \in \mathrm{S}(\mathrm{A})} e^{- \frac{|\Gamma|^2}{2{\sigma}^2 ||\mathrm{A}||^2}} d\Gamma} \cdot \frac{SW_{OPT(\mathrm{A}+\Gamma)}}{SW_{RP(\mathrm{A}+\Gamma)}} d\Gamma \\
&= \int_{\mathrm{A}+\Gamma \not\in \mathrm{S}_{c}(\mathrm{A})} \frac{e^{- \frac{|\Gamma|^2}{2{\sigma}^2 ||\mathrm{A}||^2}} \cdot \frac{SW_{OPT(\mathrm{A}+\Gamma)}}{SW_{RP(\mathrm{A}+\Gamma)}}}{\int_{\mathrm{A}+\Gamma \in \mathrm{S}(\mathrm{A})} e^{- \frac{|\Gamma|^2}{2{\sigma}^2 ||\mathrm{A}||^2}} d\Gamma}  d\Gamma
+ \int_{\mathrm{A}+\Gamma \in \mathrm{S}_{c}(\mathrm{A})} \frac{e^{- \frac{|\Gamma|^2}{2{\sigma}^2 ||\mathrm{A}||^2}} \cdot \frac{SW_{OPT(\mathrm{A}+\Gamma)}}{SW_{RP(\mathrm{A}+\Gamma)}}}{\int_{\mathrm{A}+\Gamma \in \mathrm{S}(\mathrm{A})} e^{- \frac{|\Gamma|^2}{2{\sigma}^2 ||\mathrm{A}||^2}} d\Gamma}  d\Gamma      \\
&\le \int_{\mathrm{A}+\Gamma \not\in \mathrm{S}_{c}(\mathrm{A})} \frac{e^{- \frac{|\Gamma|^2}{2{\sigma}^2 ||\mathrm{A}||^2}} \cdot n^{1-c}}{\int_{A+\Gamma \in \mathrm{S}(\mathrm{A})} e^{- \frac{|\Gamma|^2}{2{\sigma}^2 ||A||^2}} d\Gamma}  d\Gamma
+ \int_{\mathrm{A}+\Gamma \in \mathrm{S}_{c}(\mathrm{A})} \frac{e^{- \frac{|\Gamma|^2}{2{\sigma}^2 ||\mathrm{A}||^2}} \cdot O(\sqrt{n})}{\int_{\mathrm{A}+\Gamma \in \mathrm{S}(\mathrm{A})} e^{- \frac{|\Gamma|^2}{2{\sigma}^2 ||\mathrm{A}||^2}} d\Gamma}  d\Gamma\\
&\le n^{1-c} +  \frac{\mathrm{V(\mathrm{S}_{c}(\mathrm{A}))}\cdot 1 \cdot O(\sqrt{n})}{\int_{\mathrm{A}+\Gamma \in \mathrm{S}(\mathrm{A})} e^{- \frac{|\Gamma|^2}{2{\sigma}^2 ||A||^2}} d\Gamma} 
\le  n^{1-c} + \frac{ \frac{e^{2n}}{\sqrt{2\pi}n} \cdot \left( \frac{2e}{n^{1-c}} \right)^{(n(n-2))} \cdot O(\sqrt{n})}{ \left( e^{- \frac{1}{2}}\cdot \sigma ||\mathrm{A}|| \right)^{n(n-2)}}  \,\,\,\,\,\,\  \text{(by Lemma \ref{lem:smoothed1})} \\
&= n^{1-c} +  \frac{e^{2n}}{\sqrt{2\pi}n} \cdot O(\sqrt{n}) \cdot\left( \frac{2e^{\frac{3}{2}}}{n^{1-c}\sigma ||\mathrm{A}||} \right)^{(n(n-2))}
\end{align*}

Now let $1-c=\frac{2}{(n-2)\log n} + \frac{ \log (2e^{ \frac{3}{2} }) }{\log n} + \frac{\log( 1/\sigma ||\mathrm{A}|| )}{\log n}$. According to Lemma~\ref{lem:matrixnorm}, we have $||\mathrm{A}||\ge 1$, combining with the assumption $\frac{1}{\sigma} < \sqrt{n}$, we get that $\frac{\log( 1/\sigma ||\mathrm{A}|| )}{\log n} \le \frac{\log( 1/\sigma  )}{\log n} < \frac{\log( \sqrt{n} )}{\log n} = \frac{1}{2}$, which means $0<1-c<\frac{1}{2}$, as $n$ approaches infinity. Therefore the value of $1-c$ is feasible. Then,
\begin{align*}
\mathop{\mathbb{E}}_{g_{ij} \sim \mathrm{N}(0,1)} \left[ \frac{SW_{\mathrm{OPT}}(\mathrm{A}+\sigma||\mathrm{A}||\mathrm{G})}{SW_{\mathrm{RP}}(\mathrm{A}+\sigma||\mathrm{A}||\mathrm{G})} \right]
&\le 2e^{\frac{3}{2}} \cdot e^{\frac{2}{n-2}} \cdot \frac{1}{\sigma||\mathrm{A}||} + \frac{O(\sqrt{n})}{\sqrt{2\pi}n} 
 \le 2e^{\frac{3}{2}} \cdot e^{\frac{2}{n-2}} \cdot \frac{1}{\sigma} + \frac{O(\sqrt{n})}{\sqrt{2\pi}n} 
\end{align*}

So, when $\sigma||\mathrm{A}|| \le 1$ and $\frac{1}{\sigma}<\sqrt{n}$, we have $\mathop{\mathbb{E}}_{g_{ij} \sim \mathrm{N}(0,1)} \left[ \frac{SW_{\mathrm{OPT}}(\mathrm{A}+\sigma||\mathrm{A}||\mathrm{G})}{SW_{\mathrm{RP}}(\mathrm{A}+\sigma||\mathrm{A}||\mathrm{G})} \right]
 < 2e^{\frac{3}{2}}\cdot \frac{1}{\sigma}$ as $n$ approaches infinity.
 
 \indent

{\bf Case 2:} $\sigma ||\mathrm{A}|| > 1$.

We first establish the following lemma to lower bound the constant $K'$ in this case. 

\begin{lemma}\label{lem:smoothed3}
For any instance $\mathrm{A}$ and Gaussian perturbation $\Gamma$, we have
\begin{equation*}
K' = \int_{\mathrm{A}+\Gamma \in \mathrm{S}(\mathrm{A})} e^{- \frac{|\Gamma|^2}{2{\sigma}^2 ||\mathrm{A}||^2}} d\Gamma
\ge  e^{- \frac{n(n-2)}{2\sigma^2||\mathrm{A}||^2}}
\end{equation*}
\end{lemma}

\begin{proof}
\begin{align}
K' &= \int_{\mathrm{A}+\Gamma \in \mathrm{S}(\mathrm{A})} e^{- \frac{|\Gamma|^2}{2{\sigma}^2 ||\mathrm{A}||^2}} d\Gamma
= \prod_{i=1,j=1}^n \int_{\mathrm{A}+\Gamma \in \mathrm{S}(\mathrm{A})} e^{- \frac{|\gamma_{ij}|^2}{2{\sigma}^2 ||\mathrm{A}||^2}} d\gamma_{ij}  \nonumber \\
&= \prod_{(i,j)\in \mathrm{H}(\mathrm{A}) }  e^{- \frac{0}{2{\sigma}^2 ||\mathrm{A}||^2}} \prod_{(i,j)\not\in \mathrm{H}(\mathrm{A}) }\int_{a_{ij}+\gamma_{ij} \in [0,1]} e^{- \frac{|\gamma_{ij}|^2}{2{\sigma}^2 ||\mathrm{A}||^2}} d\gamma_{ij}  \nonumber \\
&= 1 \cdot  \prod_{(i,j)\not\in \mathrm{H}(\mathrm{A}) }\int_{-a_{ij}}^{1-a_{ij}} e^{- \frac{|\gamma_{ij}|^2}{2{\sigma}^2 ||\mathrm{A}||^2}} d\gamma_{ij}  \nonumber \\
&= \prod_{(i,j)\not\in \mathrm{H}(\mathrm{A}) }\left( \int_{0}^{1-a_{ij}} e^{- \frac{|\gamma_{ij}|^2}{2{\sigma}^2 ||\mathrm{A}||^2}} d\gamma_{ij} + \int_{0}^{a_{ij}} e^{- \frac{|\gamma_{ij}|^2}{2{\sigma}^2 ||\mathrm{A}||^2}} d\gamma_{ij} \right) \nonumber \\
&\ge  \prod_{(i,j)\not\in \mathrm{H}(\mathrm{A}) }\left( \int_{0}^{1-a_{ij}} e^{- \frac{|\gamma_{ij}|^2}{2{\sigma}^2 ||\mathrm{A}||^2}} d\gamma_{ij} + \int_{1-a_{ij}}^{1} e^{- \frac{|\gamma_{ij}|^2}{2{\sigma}^2 ||\mathrm{A}||^2}} d\gamma_{ij} \right) \nonumber \\
&= \prod_{(i,j)\not\in \mathrm{H}(\mathrm{A}) }\int_{0}^{1} e^{- \frac{|\gamma_{ij}|^2}{2{\sigma}^2 ||\mathrm{A}||^2}} d\gamma_{ij}  \nonumber 
= \left( \int_{0}^{1} e^{- \frac{|\gamma_{ij}|^2}{2{\sigma}^2 ||\mathrm{A}||^2}} d\gamma_{ij} \right)^{n(n-2)} \nonumber \\
&\ge \left( \int_{0}^{1} e^{- \frac{1}{2{\sigma}^2 ||\mathrm{A}||^2}} d\gamma_{ij} \right)^{n(n-2)} \nonumber 
=  e^{- \frac{n(n-2)}{2\sigma^2||\mathrm{A}||^2}} \nonumber
\end{align}
\end{proof}



Now we are ready to prove the main theorem under this case. We need to use Lemma \ref{lem:smoothed3} and handle the parameter $c$ differently.

{\bf Proof of Case 2 of Theorem \ref{thm:smoothed}}:\\
\begin{align*}
&\mathop{\mathbb{E}}_{g_{ij} \sim \mathrm{N}(0,1)} \left[ \frac{SW_{\mathrm{OPT}}(\mathrm{A}+\sigma||\mathrm{A}||\mathrm{G})}{SW_{\mathrm{RP}}(\mathrm{A}+\sigma||\mathrm{A}||\mathrm{G})} \right] 
=  \int_{\mathrm{A}+\Gamma \in \mathrm{S}(\mathrm{A})} \frac{e^{- \frac{|\Gamma|^2}{2{\sigma}^2 ||\mathrm{A}||^2}}}{\int_{\mathrm{A}+\Gamma \in \mathrm{S}(\mathrm{A})} e^{- \frac{|\Gamma|^2}{2{\sigma}^2 ||\mathrm{A}||^2}} d\Gamma} \cdot \frac{SW_{OPT(\mathrm{A}+\Gamma)}}{SW_{RP(\mathrm{A}+\Gamma)}} d\Gamma \\
&= \int_{\mathrm{A}+\Gamma \not\in \mathrm{S}_{c}(\mathrm{A})} \frac{e^{- \frac{|\Gamma|^2}{2{\sigma}^2 ||\mathrm{A}||^2}} \cdot \frac{SW_{OPT(\mathrm{A}+\Gamma)}}{SW_{RP(\mathrm{A}+\Gamma)}}}{\int_{A+\Gamma \in \mathrm{S}(\mathrm{A})} e^{- \frac{|\Gamma|^2}{2{\sigma}^2 ||\mathrm{A}||^2}} d\Gamma}  d\Gamma
+ \int_{\mathrm{A}+\Gamma \in \mathrm{S}_{c}(\mathrm{A})} \frac{e^{- \frac{|\Gamma|^2}{2{\sigma}^2 ||\mathrm{A}||^2}} \cdot \frac{SW_{OPT(\mathrm{A}+\Gamma)}}{SW_{RP(\mathrm{A}+\Gamma)}} }{\int_{A+\Gamma \in \mathrm{S}(\mathrm{A})} e^{- \frac{|\Gamma|^2}{2{\sigma}^2 ||\mathrm{A}||^2}} d\Gamma} d\Gamma\\
&\le \int_{\mathrm{A}+\Gamma \not\in \mathrm{S}_{c}(\mathrm{A})} \frac{e^{- \frac{|\Gamma|^2}{2{\sigma}^2 ||\mathrm{A}||^2}} \cdot n^{1-c}}{\int_{\mathrm{A}+\Gamma \in \mathrm{S}(\mathrm{A})} e^{- \frac{|\Gamma|^2}{2{\sigma}^2 ||\mathrm{A}||^2}} d\Gamma}  d\Gamma
+ \int_{\mathrm{A}+\Gamma \in \mathrm{S}_{c}(\mathrm{A})} \frac{e^{- \frac{|\Gamma|^2}{2{\sigma}^2 ||\mathrm{A}||^2}} \cdot O(\sqrt{n})}{\int_{\mathrm{A}+\Gamma \in \mathrm{S}(\mathrm{A})} e^{- \frac{|\Gamma|^2}{2{\sigma}^2 ||\mathrm{A}||^2}} d\Gamma}  d\Gamma\\
&\le n^{1-c} +  \frac{\mathrm{V(\mathrm{S}_{c}(\mathrm{A}))}\cdot 1 \cdot O(\sqrt{n})}{\int_{\mathrm{A}+\Gamma \in \mathrm{S}(\mathrm{A})} e^{- \frac{|\Gamma|^2}{2{\sigma}^2 ||\mathrm{A}||^2}} d\Gamma} 
\le  n^{1-c} + \frac{ \frac{e^{2n}}{\sqrt{2\pi}n} \cdot \left( \frac{2e}{n^{1-c}} \right)^{(n(n-2))} \cdot O(\sqrt{n})}{ e^{- \frac{n(n-2)}{2\sigma^2||\mathrm{A}||^2}}} \,\,\,\,\,\,\  \text{(by Lemma \ref{lem:smoothed3})} \\
&= n^{1-c} +  \frac{e^{2n}}{\sqrt{2\pi}n} \cdot O(\sqrt{n}) \cdot\left( \frac{2e^{1 + \frac{1}{2\sigma^2||\mathrm{A}||^2}}}{n^{1-c}} \right)^{(n(n-2))}
\end{align*}
Let $1-c= \frac{2}{(n-2)\log n} + \frac{1}{2\sigma^2||\mathrm{A}||^2 \log n} + \frac{\log (2e)}{\log n }$, then,

\begin{align*}
\mathop{\mathbb{E}}_{g_{ij} \sim \mathrm{N}(0,1)} \left[ \frac{SW_{\mathrm{OPT}}(\mathrm{A}+\sigma||\mathrm{A}||\mathrm{G})}{SW_{\mathrm{RP}}(\mathrm{A}+\sigma||\mathrm{A}||\mathrm{G})} \right]
&\le 2e \cdot e^{\frac{2}{n-2}+ \frac{1}{2\sigma^2||\mathrm{A}||^2}} + \frac{O(\sqrt{n})}{\sqrt{2\pi}n}
 \le  2e^{\frac{3}{2}} \cdot e^{\frac{2}{n-2}}  + \frac{O(\sqrt{n})}{\sqrt{2\pi}n} 
 \end{align*}

So, when $\sigma||\mathrm{A}|| >1$ and $\frac{1}{\sigma}<\sqrt{n}$, we have
$\mathop{\mathbb{E}}_{g_{ij} \sim \mathrm{N}(0,1)} \left[ \frac{SW_{\mathrm{OPT}}(\mathrm{A}+\sigma||\mathrm{A}||\mathrm{G})}{SW_{\mathrm{RP}}(\mathrm{A}+\sigma||\mathrm{A}||\mathrm{G})} \right]
< 2e^{\frac{3}{2}}$ as $n$ approaches infinity.

Combining these two cases we complete our proof of Theorem \ref{thm:smoothed}.


\section{Average-case Analysis}\label{sec:average}

In this section we show the average-case approximation ratio of the random priority mechanism can be improved to $1+e$. 

Firstly, we upper bound the probability of the social welfare of random priority when it is smaller than a value $n^c$, for a certain parameter $c$.

\begin{lemma}\label{lem:average}
Given any $n$, for every $0<c<1$, such that $\frac{\mathrm{e}^{2}}{\sqrt{2\pi n}}\cdot \left(\frac{2\mathrm{e}}{n^{1-c}}\right)^{n-2} < 1$, we have
\begin{equation*}
\Pr[SW_{RP}(\mathrm{A}) \le n^{c}] \le \frac{\mathrm{e}^{2}}{\sqrt{2\pi}} \cdot \sqrt{n} \cdot \left(\frac{2\mathrm{e}}{n^{1-c}}\right)^{n-2}.
\end{equation*}
\end{lemma}

\begin{proof}
For any agent $i$, since $a_{ij}, j=1,\ldots,n$, are independent and identically distributed $\mathrm{U}(0, 1)$ random variables,
the sum $\sum_{j=1}^{n}a_{ij}$ follows the {\em Irwin-Hall distribution}. In addition, in the unit-range setting, there exist $j_1$ and $j_2$ such that $a_{ij_1}=1, a_{ij_2}=0$, so $\sum\nolimits_{j=1}^{n}a_{ij}\le x$ is equivalent to $\sum\nolimits_{j \neq j_1,j_2} a_{ij}\le x-1$, given any $x>0$. Therefore,
\begin{equation*}
\Pr\left[\sum\nolimits_{j=1}^{n}a_{ij}\le x\right] = \Pr\left[\sum\nolimits_{j \neq j_1,j_2} a_{ij}\le x-1\right] = \frac{1}{(n-2)!}\sum_{k=0}^{\lfloor x-1\rfloor}(-1)^{k}\binom{n-2}{k}(x-1-k)^{n-2}
\end{equation*}
where $\lfloor \cdot \rfloor$ is the floor function. So,
\begin{align}
& \Pr\left[\sum_{j=1}^{n}a_{ij}\leqslant n^{c}\right] =\frac{1}{(n-2)!}\sum_{k=0}^{\lfloor n^{c}-1\rfloor}(-1)^{k}\binom{n-2}{k}(n^{c}-1-k)^{n-2} \nonumber \\
& \le \frac{1}{(n-2)!}\sum_{k=0}^{\lfloor n^{c}-1\rfloor}\binom{n-2}{k}(n^{c})^{n-2} 
 \le \frac{1}{(n-2)!}\sum_{k=0}^{n-2}\binom{n-2}{k}n^{c(n-2)} \nonumber  \\
& = \frac{1}{(n-2)!}\cdot (2n^{c})^{n-2} 
 \le \frac{(2n^{c})^{n-2}}{\sqrt{2\pi(n-2)}\cdot \left(\frac{n-2}{\mathrm{e}}\right)^{n-2}} \label{Stirling} \\
& =\frac{(2n^{c})^{n-2}}{\sqrt{2\pi n}\cdot \left(\frac{n}{\mathrm{e}}\right)^{n-2}}\cdot \frac{\sqrt{2\pi n}\cdot \left(\frac{n}{\mathrm{e}}\right)^{n-2}}{\sqrt{2\pi(n-2)}\cdot \left(\frac{n-2}{\mathrm{e}}\right)^{n-2}} \nonumber
 = \frac{1}{\sqrt{2\pi n}}\cdot \left(\frac{2\mathrm{e}}{n^{1-c}}\right)^{n-2}\cdot \left(1+\frac{2}{n-2}\right)^{n-2+\frac{1}{2}} \\ \nonumber
& \le \frac{1}{\sqrt{2\pi n}}\cdot \left(\frac{2\mathrm{e}}{n^{1-c}}\right)^{n-2}\cdot \mathrm{e}^{2}
 = \frac{\mathrm{e}^{2}}{\sqrt{2\pi n}}\cdot \left(\frac{2\mathrm{e}}{n^{1-c}}\right)^{n-2} \nonumber
\end{align}
where inequality (\ref{Stirling}) is due to Stirling's formula.

As we have shown in Lemma~\ref{lem:smoothed2}, in random priority mechanism,  agent $i$'s utility $u_{i}(\mathrm{A}) = \sum\nolimits_{j=1}^n a_{ij}x_{ij} \ge \frac{1}{n} \cdot \sum_{j=1}^{n}a_{ij}$. It implies that
$\Pr\left[u_{i}(\mathrm{A})>x\right] \ge \Pr\left[\frac{1}{n}\sum\nolimits_{j=1}^{n}a_{ij}>x\right]$. 
Therefore,
\begin{align}
& \Pr\left[SW_{RP}(\mathrm{A})> n^{c}\right]
=  \Pr\left[\sum_{i=1}^n u_{i}(\mathrm{A})> n^{c}\right] 
\ge \prod_{i=1}^{n}\Pr\left[u_{i}(\mathrm{A})>\frac{1}{n}\cdot n^{c}\right]  \nonumber \\
&\ge \prod_{i=1}^{n}\Pr\left[\frac{1}{n}\sum_{j=1}^{n}a_{ij}>\frac{1}{n}\cdot n^{c}\right] 
= \prod_{i=1}^{n}\Pr\left[\sum_{j=1}^{n}a_{ij}> n^{c}\right] 
=\prod_{i=1}^{n}\left(1-\Pr\left[\sum_{j=1}^{n}a_{ij}\leqslant n^{c}\right]\right) \nonumber \\
& \ge \left(1-\frac{\mathrm{e}^{2}}{\sqrt{2\pi n}}\cdot \left(\frac{2\mathrm{e}}{n^{1-c}}\right)^{n-2}\right)^{n} 
 \ge 1-n\cdot \frac{\mathrm{e}^{2}}{\sqrt{2\pi n}}\cdot \left(\frac{2\mathrm{e}}{n^{1-c}}\right)^{n-2} \label{Ber} \\
&= 1-\frac{\mathrm{e}^{2}}{\sqrt{2\pi}}\cdot\sqrt{n} \cdot \left(\frac{2\mathrm{e}}{n^{1-c}}\right)^{n-2} \nonumber
\end{align}
where inequality (\ref{Ber}) is due to Bernoulli's inequality and the condition that $\frac{\mathrm{e}^{2}}{\sqrt{2\pi n}}\cdot \left(\frac{2\mathrm{e}}{n^{1-c}}\right)^{n-2} < 1$. Hence,
\begin{equation*}
\Pr[SW_{RP}(\mathrm{A}) \le n^{c}] \le \frac{\mathrm{e}^{2}}{\sqrt{2\pi}} \cdot \sqrt{n} \cdot \left(\frac{2\mathrm{e}}{n^{1-c}}\right)^{n-2}.
\end{equation*}
\end{proof}
\vspace{-5mm}
Secondly, we lower bound the probability of the social welfare of random priority when it is lager than $\frac{n}{2}$.

\begin{lemma}\label{lem:half}
Given any $n$, we have
$\Pr[SW_{RP}(\mathrm{A}) \ge \frac{n}{2}] \ge \frac{1}{2}$.
\end{lemma}

\begin{proof}
Firstly, as we show in Lemma~\ref{lem:smoothed2}, the utility of agent $i$ in any instance $\mathrm{A}$ is $u_{i}(\mathrm{A}) = \sum\nolimits_{j=1}^n a_{ij}x_{ij} \ge \frac{1}{n} \cdot \sum_{j=1}^{n}a_{ij}$, and the social welfate $\mathrm{SW}_{\mathrm{RP}(\mathrm{A})}=\sum_i u_i(\mathrm{A}) \ge \frac{1}{n} \sum_{i,j} a_{ij}$. Therefore, event $\left\{ \sum_{i,j} a_{ij} \ge \frac{n^2}{2} \right\}$ implies $\left\{ \mathrm{SW}_{\mathrm{RP}(\mathrm{A})} \ge \frac{n}{2} \right\} $. So, we have $\Pr\left[ \sum_{i,j} a_{ij} \ge \frac{n^2}{2} \right] \le \Pr\left[\mathrm{SW}_{\mathrm{RP}}(\mathrm{A}) \ge \frac{n}{2}\right]$.

Secondly, let $\mathrm{E}$ be the $n\times n$ all-ones matrix. Denote matrix $\mathrm{A}' = (a'_{ij})_{n\times n} = \mathrm{E} - \mathrm{A}$. Obviously $\sum_{i,j} a_{ij} + \sum_{i,j} a'_{ij} = n^2$. Since all $a_{ij}$'s follow the uniform distribution, 
\begin{equation*}
\Pr\left[ \sum_{i,j} a_{ij} \ge \frac{n^2}{2} \right] = \Pr\left[ \sum_{i,j} a'_{ij} \ge \frac{n^2}{2} \right] = \Pr\left[ n^2 - \sum_{i,j} a_{ij} \ge \frac{n^2}{2} \right] = \Pr\left[ \sum_{i,j} a_{ij} \le \frac{n^2}{2} \right] .
\end{equation*}
So $\Pr\left[ \sum_{i,j} a_{ij} \ge \frac{n^2}{2} \right] = \frac{1}{2}$. Hence, $\Pr\left[\mathrm{SW}_{\mathrm{RP}}(\mathrm{A}) \ge \frac{n}{2}\right] \ge \Pr\left[ \sum_{i,j} a_{ij} \ge \frac{n^2}{2} \right] = \frac{1}{2}$.
\end{proof}

Next we will use Lemma \ref{lem:average} and Lemma~\ref{lem:half} to prove our main result in this section. Essentially, Lemma \ref{lem:average} and Lemma~\ref{lem:half} bound the probability of the social welfare of random priority on valuation profile $\mathrm{A}$. By carefully choosing parameter $c$, we can divide the valuation space into three sets: $\{SW_{RP(\mathrm{A})}\le n^c\}$, $\{n^{c}<SW_{RP(\mathrm{A})}<\frac{n}{2}\}$ and $\{SW_{RP(\mathrm{A})}\ge\frac{n}{2}\}$. We bound the probabilities of instance $\mathrm{A}$ falling into each set and the ratios of optimal social welfare against the social welfare of random priority mechanism. Note that in any cases, the worst-case ratio is upper bounded by $O(\sqrt{n})$. By adding them up together, we obtain our upper bound of $1+e$.
\begin{theorem}\label{thm:average}
The average case approximation ratio is upper bounded by  $1+e$. That is,
\begin{equation*}
r_{average} = \mathop{\mathbb{E}}_{a_{ij}\sim \mathrm{U}} \left[ \frac{SW_{OPT(\mathrm{A})}}{SW_{RP(\mathrm{A})}} \right] \le 1+e.
\end{equation*}
\end{theorem}

\begin{proof}
Let $1-c = \log_{n} (2e) + \frac{2}{n-2}$. It is easy to verify that $c$  satisfies the condition of Lemma \ref{lem:average} and $n^c < \frac{n}{2}$. So the above three sets are collectively exhaustive and mutually exclusive, and we have
\begin{align*}
& \mathop{\mathbb{E}}_{a_{ij}\sim \mathrm{U}} \left[ \frac{SW_{OPT(\mathrm{A})}}{SW_{RP(\mathrm{A})}} \right]
= \Pr\left[SW_{RP(\mathrm{A})}\ge\frac{n}{2}\right] \cdot \mathop{\mathbb{E}}_{_{SW_{RP(\mathrm{A})}\ge\frac{n}{2}}} \left[ \frac{SW_{OPT(\mathrm{A})}}{SW_{RP(\mathrm{A})}} \right] + \\
&\Pr\left[n^c<SW_{RP(\mathrm{A})}<\frac{n}{2}\right] \cdot \mathop{\mathbb{E}}_{_{n^c<SW_{RP(\mathrm{A})}<\frac{n}{2}}} \left[ \frac{SW_{OPT(\mathrm{A})}}{SW_{RP(\mathrm{A})}} \right]+ 
\Pr\left[SW_{RP(\mathrm{A})}\le n^c\right] \cdot \mathop{\mathbb{E}}_{_{SW_{RP(\mathrm{A})}\le n^c}} \left[ \frac{SW_{OPT(\mathrm{A})}}{SW_{RP(\mathrm{A})}} \right] \\
&\le \Pr\left[SW_{RP(\mathrm{A})}\ge\frac{n}{2}\right] \cdot \frac{n}{n/2}+
(1-\Pr\left[SW_{RP(\mathrm{A})}\ge\frac{n}{2}\right]) \cdot \frac{n}{n^c} + \Pr\left[SW_{RP(\mathrm{A})}\le n^c\right] \cdot O(\sqrt{n}) \\
&<n^{1-c} + (2-n^{1-c})\cdot\Pr\left[SW_{RP(\mathrm{A})}\ge\frac{n}{2}\right]+\frac{e^2}{\sqrt{2\pi}} \cdot \sqrt{n} \cdot \left( \frac{2e}{n^{1-c}} \right)^{n-2} \cdot O(\sqrt{n})
\end{align*}

By our choice of $c$, we have $n^{1-c}=2e\cdot n^{\frac{2}{n-2}}$. So,
\begin{align*}
r_{average} 
&< 2e \cdot n^{\frac{2}{n-2}} - (2e \cdot n^{\frac{2}{n-2}}-2)\cdot\Pr\left[SW_{RP(\mathrm{A})}\ge\frac{n}{2}\right] + \frac{e^2}{\sqrt{2\pi}} \cdot \sqrt{n} \cdot \left( \frac{2e}{2e \cdot n^{\frac{2}{n-2}}} \right)^{n-2} \cdot O(\sqrt{n})\\
&\le 2e \cdot n^{\frac{2}{n-2}} - (2e \cdot n^{\frac{2}{n-2}}-2)\cdot \frac{1}{2}+\frac{e^2}{\sqrt{2\pi}}  \cdot \frac{1}{n^2} \cdot O(n) \\
&=e\cdot n^{\frac{2}{n-2}} +1+ \frac{e^2}{\sqrt{2\pi}}  \cdot \frac{1}{n^2} \cdot O(n) <1+e.
\end{align*}
\end{proof}
\vspace{-5mm}
\section{Conclusion and Discussion}\label{sec:discussion}

In this paper, we extend the worst-case approximation ratio analysis in mechanism design to the smoothed approximation ratio and average-case approximation ratio analysis. For  social welfare maximization in one-sided matching problem, we show a clear separation of the approximation ratio bounds from the  smoothed analysis and average-case analysis to the worst-case analysis. Our results partially explain why random priority has been successfully  used in practice, although in the worst case the optimal social welfare is $\Theta{(\sqrt{n})}$ times of what random priority achieves.


There are quite a few emerging open questions in the smoothed and average-case approximation ratio analysis of mechanisms. Firstly, it would be good to improve our upper bounds and to characterize matching lower bound. Secondly, our average-case analysis is based on a uniform distribution; it is open to consider other distributions that resembles real-life applications of the one-sided matching mechanisms. Thirdly, beside unit-range representation, another interesting valuation normalization is unit-sum; it is open to study the smoothed  and average-case analysis in that setting. Last but most importantly, our analysis pave the way of characterizing the smoothed and average-case  approximation ratio in other mechanism design problems, such as the scheduling problem.

\bibliographystyle{plainnat}
\bibliography{refs}

\begin{thebibliography}{51}
\providecommand{\natexlab}[1]{#1}
\providecommand{\url}[1]{\texttt{#1}}
\expandafter\ifx\csname urlstyle\endcsname\relax
  \providecommand{\doi}[1]{doi: #1}\else
  \providecommand{\doi}{doi: \begingroup \urlstyle{rm}\Url}\fi

\bibitem[Abdulkadiro{\u{g}}lu and S{\"o}nmez(1998)]{abdulkadirouglu1998random}
Atila Abdulkadiro{\u{g}}lu and Tayfun S{\"o}nmez.
\newblock Random serial dictatorship and the core from random endowments in
  house allocation problems.
\newblock \emph{Econometrica}, pages 689--701, 1998.

\bibitem[Abdulkadiro{\u{g}}lu and S{\"o}nmez(2013)]{AS:13}
Atila Abdulkadiro{\u{g}}lu and Tayfun S{\"o}nmez.
\newblock {Matching Markets: Theory and Practice}.
\newblock \emph{Advances in Economics and Econometrics (Tenth World Congress)},
  pages 3--47, 2013.

\bibitem[Alaei et~al.(2015)Alaei, Hartline, Niazadeh, Pountourakis, and
  Yuan]{DBLP:conf/focs/AlaeiHNPY15}
Saeed Alaei, Jason~D. Hartline, Rad Niazadeh, Emmanouil Pountourakis, and Yang
  Yuan.
\newblock Optimal auctions vs. anonymous pricing.
\newblock In \emph{{IEEE} 56th Annual Symposium on Foundations of Computer
  Science, {FOCS}}, pages 1446--1463, 2015.

\bibitem[Archer et~al.(2003)Archer, Papadimitriou, Talwar, and
  Tardos]{DBLP:conf/soda/ArcherPTT03}
Aaron Archer, Christos~H. Papadimitriou, Kunal Talwar, and {\'{E}}va Tardos.
\newblock An approximate truthful mechanism for combinatorial auctions with
  single parameter agents.
\newblock In \emph{Proceedings of the Fourteenth Annual {ACM-SIAM} Symposium on
  Discrete Algorithms, SODA}, pages 205--214, 2003.

\bibitem[Arthur and Vassilvitskii(2009)]{DBLP:journals/siamcomp/ArthurV09}
David Arthur and Sergei Vassilvitskii.
\newblock Worst-case and smoothed analysis of the {ICP} algorithm, with an
  application to the k-means method.
\newblock \emph{{SIAM} J. Comput.}, 39\penalty0 (2):\penalty0 766--782, 2009.

\bibitem[Ashlagi et~al.(2010)Ashlagi, Fischer, Kash, and
  Procaccia]{ashlagi2010mix}
I.~Ashlagi, F.~Fischer, I.~Kash, and Ariel~D. Procaccia.
\newblock Mix and match.
\newblock In \emph{Proceedings of the 11th ACM conference on Electronic
  commerce (ACM-EC)}, pages 305--314, 2010.

\bibitem[Banderier et~al.(2003)Banderier, Beier, and
  Mehlhorn]{DBLP:conf/mfcs/BanderierBM03}
Cyril Banderier, Ren{\'{e}} Beier, and Kurt Mehlhorn.
\newblock Smoothed analysis of three combinatorial problems.
\newblock In \emph{28th International Symposium, Mathematical Foundations of
  Computer Science, {MFCS}}, pages 198--207, 2003.

\bibitem[Barbera(2010)]{Barbera10}
Salvador Barbera.
\newblock {Strategy-proof Social Choice}.
\newblock In K.~J. Arrow, A.~K. Sen, and K.~Suzumura, editors, \emph{Handbook
  of Social Choice and Welfare}, volume~2, chapter~25. North-Holland:
  Amsterdam, 2010.

\bibitem[Becchetti et~al.(2005)Becchetti, Leonardi, Marchetti{-}Spaccamela,
  Sch{\"{a}}fer, and Vredeveld]{DBLP:conf/dagstuhl/BecchettiLMSV05}
Luca Becchetti, Stefano Leonardi, Alberto Marchetti{-}Spaccamela, Guido
  Sch{\"{a}}fer, and Tjark Vredeveld.
\newblock Average case and smoothed competitive analysis of the multi-level
  feedback algorithm.
\newblock In \emph{Algorithms for Optimization with Incomplete Information,
  pages 16--21}, 2005.

\bibitem[Bl{\"{a}}ser et~al.(2013)Bl{\"{a}}ser, Manthey, and
  Rao]{DBLP:journals/algorithmica/BlaserMR13}
Markus Bl{\"{a}}ser, Bodo Manthey, and B.~V.~Raghavendra Rao.
\newblock Smoothed analysis of partitioning algorithms for euclidean
  functionals.
\newblock \emph{Algorithmica}, 66\penalty0 (2):\penalty0 397--418, 2013.

\bibitem[Blum and Dunagan(2002)]{DBLP:conf/soda/BlumD02}
Avrim Blum and John Dunagan.
\newblock Smoothed analysis of the perceptron algorithm for linear programming.
\newblock In David Eppstein, editor, \emph{Proceedings of the Thirteenth Annual
  {ACM-SIAM} Symposium on Discrete Algorithms, SODA}, pages 905--914, 2002.

\bibitem[Bogomolnaia and Moulin(2001)]{BM:01}
Anna Bogomolnaia and Herv\'{e} Moulin.
\newblock {A New Solution to the Random Assignment Problem}.
\newblock \emph{Journal of Economic Theory}, 100:\penalty0 295--328, 2001.

\bibitem[Borgwardt(1982)]{DBLP:journals/mmor/Borgwardt82}
Karl{-}Heinz Borgwardt.
\newblock The average number of pivot steps required by the simplex-method is
  polynomial.
\newblock \emph{Zeitschr. f{\"{u}}r {OR}}, 26\penalty0 (1):\penalty0 157--177,
  1982.

\bibitem[Boros et~al.(2011)Boros, Elbassioni, Fouz, Gurvich, Makino, and
  Manthey]{DBLP:conf/icalp/BorosEFGMM11}
Endre Boros, Khaled~M. Elbassioni, Mahmoud Fouz, Vladimir Gurvich, Kazuhisa
  Makino, and Bodo Manthey.
\newblock Stochastic mean payoff games: Smoothed analysis and approximation
  schemes.
\newblock In \emph{Automata, Languages and Programming - 38th International
  Colloquium, {ICALP} Part {I}}, pages 147--158, 2011.

\bibitem[Chawla and Sivan(2014)]{DBLP:journals/sigecom/ChawlaS14}
Shuchi Chawla and Balasubramanian Sivan.
\newblock Bayesian algorithmic mechanism design.
\newblock \emph{SIGecom Exchanges}, 13\penalty0 (1):\penalty0 5--49, 2014.

\bibitem[Chen et~al.(2009)Chen, Deng, and Teng]{DBLP:journals/jacm/ChenDT09}
Xi~Chen, Xiaotie Deng, and Shang{-}Hua Teng.
\newblock Settling the complexity of computing two-player nash equilibria.
\newblock \emph{J. {ACM}}, 56\penalty0 (3), 2009.

\bibitem[Christodoulou et~al.(2009)Christodoulou, Koutsoupias, and
  Vidali]{DBLP:journals/algorithmica/ChristodoulouKV09}
George Christodoulou, Elias Koutsoupias, and Angelina Vidali.
\newblock A lower bound for scheduling mechanisms.
\newblock \emph{Algorithmica}, 55\penalty0 (4):\penalty0 729--740, 2009.

\bibitem[Cormen et~al.(2009)Cormen, Leiserson, Rivest, and
  Stein]{DBLP:books/daglib/0023376}
Thomas~H. Cormen, Charles~E. Leiserson, Ronald~L. Rivest, and Clifford Stein.
\newblock \emph{Introduction to Algorithms {(3.} ed.)}.
\newblock {MIT} Press, 2009.

\bibitem[Dughmi and Ghosh(2010)]{DG:10}
Shaddin Dughmi and Arpita Ghosh.
\newblock {Truthful assignment without money}.
\newblock In \emph{ACM Conference on Electronic Commerce}, pages 325--334,
  2010.

\bibitem[Englert et~al.(2014)Englert, R{\"{o}}glin, and
  V{\"{o}}cking]{DBLP:journals/algorithmica/EnglertRV14}
Matthias Englert, Heiko R{\"{o}}glin, and Berthold V{\"{o}}cking.
\newblock Worst case and probabilistic analysis of the 2-opt algorithm for the
  {TSP}.
\newblock \emph{Algorithmica}, 68\penalty0 (1):\penalty0 190--264, 2014.

\bibitem[Filos-Ratsikas et~al.(2014)Filos-Ratsikas, Frederiksen, and
  Zhang]{FFZ:14}
Aris Filos-Ratsikas, S{\o}ren Kristoffer~Stiil Frederiksen, and Jie Zhang.
\newblock Social welfare in one-sided matchings: Random priority and beyond.
\newblock In \emph{7th International Symposium on Algorithmic Game Theory
  (SAGT)}, pages 1--12, 2014.

\bibitem[Filos{-}Ratsikas et~al.(2015)Filos{-}Ratsikas, Li, Zhang, and
  Zhang]{DBLP:conf/aaai/Filos-RatsikasL15}
Aris Filos{-}Ratsikas, Minming Li, Jie Zhang, and Qiang Zhang.
\newblock Facility location with double-peaked preferences.
\newblock In \emph{Proceedings of the Twenty-Ninth {AAAI} Conference on
  Artificial Intelligence}, pages 893--899, 2015.

\bibitem[Fouz et~al.(2012)Fouz, Kufleitner, Manthey, and
  Jahromi]{DBLP:journals/algorithmica/FouzKMJ12}
Mahmoud Fouz, Manfred Kufleitner, Bodo Manthey, and Nima~Zeini Jahromi.
\newblock On smoothed analysis of quicksort and hoare's find.
\newblock \emph{Algorithmica}, 62\penalty0 (3-4):\penalty0 879--905, 2012.

\bibitem[Guo and Conitzer(2010)]{guo2010strategyproof}
M.~Guo and V.~Conitzer.
\newblock Strategy-proof allocation of multiple items between two agents
  without payments or priors.
\newblock In \emph{Ninth International Joint Conference on Autonomous Agents
  and Multi Agent Systems (AAMAS)}, volume~10, pages 881--888, 2010.

\bibitem[Hartline and Lucier(2010)]{DBLP:conf/stoc/HartlineL10}
Jason~D. Hartline and Brendan Lucier.
\newblock Bayesian algorithmic mechanism design.
\newblock In \emph{Proceedings of the 42nd {ACM} Symposium on Theory of
  Computing, {STOC} 2010}, pages 301--310, 2010.

\bibitem[Hartline and Roughgarden(2009)]{DBLP:conf/sigecom/HartlineR09}
Jason~D. Hartline and Tim Roughgarden.
\newblock Simple versus optimal mechanisms.
\newblock In \emph{Proceedings 10th {ACM} Conference on Electronic Commerce,
  ACM-EC}, pages 225--234, 2009.

\bibitem[Hylland and Zeckhauser(1979)]{HZ:79}
Aanund Hylland and Richard Zeckhauser.
\newblock {The Efficient Allocation of Individuals to Positions}.
\newblock \emph{The Journal of Political Economy}, 87\penalty0 (2):\penalty0
  293--314, 1979.

\bibitem[Koutsoupias(2014)]{DBLP:journals/mst/Koutsoupias14}
Elias Koutsoupias.
\newblock Scheduling without payments.
\newblock \emph{Theory Comput. Syst.}, 54\penalty0 (3):\penalty0 375--387,
  2014.

\bibitem[Koutsoupias and Vidali(2007)]{DBLP:conf/mfcs/KoutsoupiasV07}
Elias Koutsoupias and Angelina Vidali.
\newblock A lower bound of 1+$\phi$ for truthful scheduling mechanisms.
\newblock In \emph{the 32nd International Symposium, Mathematical Foundations
  of Computer Science, MFCS}, pages 454--464, 2007.

\bibitem[Lehmann et~al.(2002)Lehmann, O'Callaghan, and
  Shoham]{DBLP:journals/jacm/LehmannOS02}
Daniel~J. Lehmann, Liadan O'Callaghan, and Yoav Shoham.
\newblock Truth revelation in approximately efficient combinatorial auctions.
\newblock \emph{J. {ACM}}, 49\penalty0 (5):\penalty0 577--602, 2002.

\bibitem[Manthey and Reischuk(2007)]{DBLP:journals/tcs/MantheyR07}
Bodo Manthey and R{\"{u}}diger Reischuk.
\newblock Smoothed analysis of binary search trees.
\newblock \emph{Theor. Comput. Sci.}, 378\penalty0 (3):\penalty0 292--315,
  2007.

\bibitem[Manthey and R{\"{o}}glin(2011)]{DBLP:journals/it/MantheyR11}
Bodo Manthey and Heiko R{\"{o}}glin.
\newblock Smoothed analysis: Analysis of algorithms beyond worst case.
\newblock \emph{it - Information Technology}, 53\penalty0 (6):\penalty0
  280--286, 2011.

\bibitem[Mennle and Seuken(2014)]{mennle2014axiomatic}
Timo Mennle and Sven Seuken.
\newblock {An axiomatic approach to characterizing and relaxing
  strategyproofness of one-sided matching mechanisms}.
\newblock In \emph{Proceedings of the 15th ACM Conference on Economics and
  Computation}, pages 37--38, 2014.

\bibitem[Mu'alem and Nisan(2008)]{DBLP:journals/geb/MualemN08}
Ahuva Mu'alem and Noam Nisan.
\newblock Truthful approximation mechanisms for restricted combinatorial
  auctions.
\newblock \emph{Games and Economic Behavior}, 64\penalty0 (2):\penalty0
  612--631, 2008.

\bibitem[Nisan and Ronen(1999)]{NR99}
Noam Nisan and Amir Ronen.
\newblock Algorithmic mechanism design.
\newblock In \emph{Proceedings of the Thirty-First Annual ACM Symposium on
  Theory of Computing, STOC}, pages 129--140, 1999.

\bibitem[Nisan and Ronen(2001)]{DBLP:journals/geb/NisanR01}
Noam Nisan and Amir Ronen.
\newblock Algorithmic mechanism design.
\newblock \emph{Games and Economic Behavior}, 35\penalty0 (1-2):\penalty0
  166--196, 2001.

\bibitem[Nisan et~al.(2007)Nisan, Roughgarden, Tardos, and Vazirani]{NRTV07}
Noam Nisan, Tim Roughgarden, Eva Tardos, and Vijay~V. Vazirani, editors.
\newblock \emph{Algorithmic Game Thoery}.
\newblock Cambridge University Press, 2007.

\bibitem[Parag A.~Pathak(2011)]{PathakSethuraman}
Jay~Sethuraman Parag A.~Pathak.
\newblock {Lotteries in student assignment: An equivalence result}.
\newblock \emph{Theoretical Economics}, 6:\penalty0 1--17, 2011.

\bibitem[Procaccia and Tennenholtz(2009)]{PT:09}
Ariel~D Procaccia and Moshe Tennenholtz.
\newblock {Approximate mechanism design without money}.
\newblock In \emph{Proceedings of the 10th ACM Conference on Electronic
  Commerce}, pages 177--186. ACM, 2009.

\bibitem[Sankar et~al.(2006)Sankar, Spielman, and
  Teng]{DBLP:journals/siammax/SankarST06}
Arvind Sankar, Daniel~A. Spielman, and Shang{-}Hua Teng.
\newblock Smoothed analysis of the condition numbers and growth factors of
  matrices.
\newblock \emph{{SIAM} J. Matrix Analysis Applications}, 28\penalty0
  (2):\penalty0 446--476, 2006.

\bibitem[Sch{\"{a}}fer and Sivadasan(2005)]{DBLP:journals/tcs/SchaferS05}
Guido Sch{\"{a}}fer and Naveen Sivadasan.
\newblock Topology matters: Smoothed competitiveness of metrical task systems.
\newblock \emph{Theor. Comput. Sci.}, 341\penalty0 (1-3):\penalty0 216--246,
  2005.

\bibitem[S{\"o}nmez and {\"U}nver(2011)]{SU:11}
Tayfun S{\"o}nmez and Utku {\"U}nver.
\newblock {Matching, allocation and exchange of discrete resources}.
\newblock \emph{Handbook of Social Economics}, 1A:\penalty0 781--852, 2011.

\bibitem[Spielman and Teng(2001)]{DBLP:conf/stoc/SpielmanT01}
Daniel~A. Spielman and Shang{-}Hua Teng.
\newblock Smoothed analysis of algorithms: why the simplex algorithm usually
  takes polynomial time.
\newblock In \emph{Proceedings on 33rd Annual {ACM} Symposium on Theory of
  Computing, STOC}, pages 296--305, 2001.

\bibitem[Spielman and Teng(2003)]{DBLP:journals/mp/SpielmanT03}
Daniel~A. Spielman and Shang{-}Hua Teng.
\newblock Smoothed analysis of termination of linear programming algorithms.
\newblock \emph{Math. Program.}, 97\penalty0 (1-2):\penalty0 375--404, 2003.

\bibitem[Spielman and Teng(2009)]{DBLP:journals/cacm/SpielmanT09}
Daniel~A. Spielman and Shang{-}Hua Teng.
\newblock Smoothed analysis: an attempt to explain the behavior of algorithms
  in practice.
\newblock \emph{Commun. {ACM}}, 52\penalty0 (10):\penalty0 76--84, 2009.

\bibitem[Svensson(1999{\natexlab{a}})]{SVE:99}
Lars-Gunnar Svensson.
\newblock {Strategy-proof allocation of indivisble goods}.
\newblock \emph{Social Choice and Welfare}, 16\penalty0 (4):\penalty0 557--567,
  1999{\natexlab{a}}.

\bibitem[Svensson(1999{\natexlab{b}})]{Sv99}
Lars-Gunnar Svensson.
\newblock Strategy-proof allocation of indivisible goods.
\newblock \emph{Social Choice and Welfare}, 16\penalty0 (4):\penalty0 557--567,
  1999{\natexlab{b}}.

\bibitem[Szpankowsk(2010)]{Szpankowsk}
Wojciech Szpankowsk.
\newblock \emph{Average case analysis of algorithms}.
\newblock Chapman Hall CRC, 2010.
\newblock ISBN 978-1-58488-822-2.

\bibitem[Von~Neumann and Morgenstern(1953)]{vNM1953}
John Von~Neumann and Oskar Morgenstern.
\newblock \emph{Theory of games and economic behavior}.
\newblock Princeton university press, 1953.

\bibitem[Von~Neumann and Morgenstern(2007)]{von2007theory}
John Von~Neumann and Oskar Morgenstern.
\newblock \emph{Theory of games and economic behavior (60th Anniversary
  Commemorative Edition)}.
\newblock Princeton university press, 2007.

\bibitem[Zhou(1990)]{ZHOU:90}
Lin Zhou.
\newblock {On a Conjecture by Gale about One-Sided Matching Problems}.
\newblock \emph{Journal of Economic Theory}, 52:\penalty0 123--135, 1990.

\end{thebibliography}

\end{document}